\documentclass[12pt]{article}
\usepackage{pictex} \usepackage{graphics} \usepackage{amsmath} 
\usepackage{amssymb} \usepackage{amsthm} 
\newcommand{\nc}{\newcommand}  
\nc{\di}{\displaystyle} \nc{\nn}{\nonumber} 
\nc{\ns}{\normalsize} \nc{\seq}{\subseteq} 
\renewcommand{\Re}{\mathrm{Re}\,} \renewcommand{\Im}{\mathrm{Im}\,} 
\nc{\AH}{\mathcal{A}} \nc{\BE}{\mathcal{B}} \nc{\CE}{\mathcal{C}} 
\nc{\D}{\mathcal{D}} \nc{\E}{\mathcal{E}} \nc{\EF}{\mathcal{F}} 
\nc{\GE}{\mathcal{G}} \nc{\HA}{\mathcal{H}} \nc{\J}{\mathcal{J}} 
\nc{\KA}{\mathcal{K}} \nc{\EL}{\mathcal{L}} \nc{\PE}{\mathcal{P}} 
\nc{\ER}{\mathcal{R}} \nc{\ES}{\mathcal{S}} \nc{\TE}{\mathcal{T}} 
\nc{\EM}{\mathcal{M}} \nc{\EN}{\mathcal{N}} \nc{\OH}{\mathcal{O}} 
\nc{\U}{\mathcal{U}} \nc{\WE}{\mathcal{W}} \nc{\EX}{\mathcal{X}} 
\nc{\Y}{\mathcal{Y}} 
\nc{\ma}[1]{\mbox{$\,{#1}\,$}} \nc{\ek}{\protect\\[1ex]} 
\nc{\zx}{\protect\\[2ex]}

 \newcommand{\R}{{\mathbb R}} 
  
 \nc{\IR}{\mbox{\bf R}} \nc{\IN}{\mbox{\bf N}} 
\nc{\ZZ}{\mbox{\bf Z}} \nc{\la}{\lambda} \nc{\La}{\Lambda} 
\nc{\da}{\delta} \nc{\Da}{\Delta} \nc{\ta}{\theta} \nc{\Ta}{\Theta} 
\nc{\na}{\nabla} \nc{\ue}{\infty} \nc{\vp}{\varphi} \nc{\vta}{\vartheta} 
\nc{\Gm}{\Gamma} \nc{\gm}{\gamma} \nc{\ka}{\kappa} \nc{\si}{\sigma} 
\nc{\Si}{\Sigma} \nc{\al}{\alpha} \nc{\be}{\beta} \nc{\om}{\omega} 
\nc{\Om}{\Omega} \nc{\pa}{\partial} \nc{\ti}{\times} \nc{\n}{|} 
\nc{\rub}{\,\rule[-2.7pt]{.02in}{4mm}\,} \nc{\ab}{\|} \nc{\s}{\tilde} 
\nc{\ve}{\varepsilon} \nc{\fa}{\forall} \nc{\ov}{\overline} 
\nc{\un}{\underline} \nc{\Llr}{\Longleftrightarrow} 
\nc{\llr}{\longleftrightarrow} \nc{\ra}{\rightarrow} 
\nc{\lra}{\longrightarrow} \nc{\rh}{\rightharpoonup} \nc{\Ra}{\Rightarrow} 
\nc{\ran}{\rangle} \nc{\lan}{\langle} \nc{\bs}{\backslash} \nc{\ko}{\,,\,} 
\nc{\eq}[1]{\mbox{\rm {(\ref{E#1})}}} 
\nc{\ha}{\frac{1}{2}} 
\nc{\kla}{\,[\,} \nc{\klz}{\,]\,} \nc{\lk}{\left[} \nc{\rk}{\right]} 
\nc{\lb}{\left\{} \nc{\rb}{\right\}} \nc{\rr}{\right)} \nc{\lr}{\left(} 
\nc{\f}{\big(} \nc{\g}{\big)} \nc{\Ba}{\Big(} \nc{\Bz}{\Big)} 
\nc{\Bka}{\Big[} \nc{\Bkz}{\Big]} \nc{\bka}{\big[} \nc{\bkz}{\big]} 
\nc{\Blb}{\Big\{} \nc{\Brb}{\Big\}} \nc{\blb}{\big\{} \nc{\brb}{\big\}} 
\nc{\pn}{\par\noindent} \nc{\emp}{\emptyset} \nc{\Ri}{\Rightarrow} 
\nc{\hph}{\hphantom} \nc{\vph}{\vphantom} 
\nc{\vpn}{\vspace{2ex}\par\noindent} \nc{\vpar}{\vspace{2ex}\par} 
\nc{\mathe}[1]{\mbox{${\di {#1}}$}} \nc{\equno}[1]{\\[-.5ex] 
\mbox{}\label{#1}\\[-.5ex]} \nc{\tr}{{\mathrm{tr}}\,} 
\newtheorem{lem}{Lemma}

\begin{document} 
\title{{\Large {\bf 
A recipe 
for making materials with negative refraction in acoustics. }}} \author{A. 
G. Ramm \\ {\ns (Mathematics Department, Kansas St. University,}\\ {\ns 
Manhattan, KS66506, USA} \\ {\ns and TU Darmstadt, Germany)}\\ {\small 
ramm@math.ksu.edu}} \date{} 
\maketitle 
\begin{abstract}\noindent A 
recipe is given for making materials with negative refraction in 
acoustics, i.e., materials in which the group velocity is directed 
opposite to the phase velocity.

The recipe consists of injecting many small particles into a bounded 
domain, filled with a material whose refraction coefficient is known. The 
number of small particles to be injected per unit volume around any point 
$x$ is calculated as well as the  boundary impedances of
 the embedded particles. 
\end{abstract}

{\small PACS:03.04.Kf   } 
\ek {\small Keywords: metamaterials, "smart" materials, wave scattering by small bodies} 

\section{Introduction }\label{S1} 
Suppose that a bounded domain $D\subset \R^3$ is filled with a material 
whose refraction coefficient $n_0^2(x)$ is known. Acoustic wave scattering 
problem consists of finding the solution to the equation
\begin{eqnarray}
\label{E1}[ \na^2 + k^2 n_0^2(x)] u & = & 0 \text{ in } \R^3,\ek
\label{E2}u=e^{ik\al \cdot x} + v;&& r\f\frac{\pa v}{\pa r}-ikv\g	
=o(1), \ r\to \infty.
\end{eqnarray}
Here $\al \in S^2$ is a given unit vector, the direction of the incident wave, $S^2$ is the unit sphere in $\R^3$, $k>0$ is the wave number, $n_0^2(x)=1$ in $D':=\R^3 \backslash D$, and $v$ is the scattered field. We assume that $\Im n_0^2(x)\ge 0$. This corresponds to possible absorption in the material. The function $u$ is the acoustic potential or pressure, and the time dependence factor $e^{-i\om t}$ is omitted. It is known that problem \eq{1} - \eq{2} has a solution and this solution is unique.

The problem we study in this paper is the following one:
\pn
\textit{Can one create in $D$ a material with negative refraction, that is, a material in which the group velocity is directed opposite to the phase velocity?}
\pn
If there is a wave-packet $u=\sum_k a(k) e^{i(k\cdot r-\om(k)t)}$, 
where $| k-\ov{k} | < \da$, $| \om(k) - \om(\ov{k})| < \da$, $\da>0$ 
is a small number, $a(k)$ is the complex amplitude, $\ov{k}$ is the
wave vector around which the other wavevectors are grouped,  
and summation over $k$ should be replaced by integration 
if vector $k$ changes continuously, then the 
group velocity of this packet is defined by the formula: 
$$v_g=\na_k\om(k).$$ 
Its phase velocity is defined by the formula: 
$$v_p=\frac{\om(k)}{| k |}k^0,$$ 
where $k^0:=\frac{k}{|k|}$,  $|k|$ is the length of the wave vector 
$k$, and
$\na_k \om(k)$ is the gradient of the scalar function $\om(k)$.

Suppose $\om(k)>0$. Then $v_p$ is directed along $k$. Vector 
$v_g=\na_k\om(k)$ is directed along $-k$ if $\om=\om(|k|)$ and 
$\om'(|k|)<0$. Indeed, 
$$\na_k \om(|k|)=\om'(|k|)k^0,$$
where the prime denotes the derivative with respect to the
argument $|k|$.

Therefore, one can produce the desired material with negative refraction 
if one can produce the material with
	\[
	\om=\om(|k|) \quad \text{and} \quad \om'(|k|)<0.
\]

The notion of negative refraction in optics is discussed in detail
in \cite{A}, where the role of spatial dispersion 
in creating materials with negative refraction is emphasized. In 
\cite{R476} a theory of wave scattering by small
bodies of arbitrary shapes is developed. On the basis of this theory
it is proved in \cite{R529} and in \cite{R528} that one can create 
material with a desired 
refraction coefficient $n^2(x)$ at a fixed frequency by embedding many 
small particles into the region $D$, filled with a material with
known refraction coefficient $n_0^2(x)$.

Let us formulate one of the results from \cite{R529}.

Suppose that the small particles are balls $B_m$, $1\le m 
\le M$, $M$ is the total number of small balls embedded in $D$,  
the radius of these balls is $a$, the distance $d$ between two neighboring 
balls is of 
order $O(a^{1/3})$, the boundary condition on the surface $S_m$ of the 
$m-$th ball is 
$$u_N=\zeta_m u \quad \hbox{\,\, on\,\, } S_m,$$ 
where  $\zeta_m:=\frac{h(x_m)}{a}$,  
 $x_m$ is the center of the $m$-th ball, and
$h(x)$ is an arbitrary 
given continuous function in $D$, with the property $\Im h(x)\le 0$. 

Assume that
$$\lim_{a \to 0} \{ a\EN(\Da)\} = \int_\Da N(x)dx$$ 
for any subdomain 
$\Da \subset D$, where $\EN(\Da)$ is the number of particles in $\Da$, 
and $N(x)\ge 0$ is an arbitrary given continuous function in $D$.
Let $[b]$ denote the nearest integer to the numbe $b\geq 0$.

{\it Under these assumptions  it 
is proved in \cite{R529} that embedding in every subcube $\Da_p \subset D$ 
the number $[\frac{1}{a}\int_\Da N(x)dx]$ particles at the distances of 
order $O(a^{1/3})$ with boundary parameters $\zeta_m=\frac{h(x_m)}{a}$, 
creates new
material in $D$ with the refraction coefficient} 
$$ n^2(x)=n_0^2(x) - k^{-2}p(x),$$ 
{\it where}
\begin{equation}
	\label{E3}
	p(x)=\frac{4\pi h(x)}{1+h(x)}N(x), \quad x \in D.
\end{equation}
  Given $p = p_1 + ip_2$, where $p_1 = {\rm Re}\, p$, $p_2 = {\rm Im}\, p
\le 0$, one finds (non-uniquely) $h_1(x)=\Re h$, $h_2(x) = \Im h(x) \leq 
0$, and
$N(x)\geq 0$, such that \eq{3} holds. In \cite{R529} it was assumed that
$k$ was fixed,
$k=\frac{\om}{c}$, $c = {\rm const}$ is the wave speed in the exterior 
homogeneous region $D'=\R^3\setminus D$. 

{\it The purpose of this paper is to point out that
the method in \cite{R529} is valid in the case when $p = p(x,\om)$.
This implies that this method allows one to create the material
with the desired spatial dispersion. In particular, one 
can create a material with negative refraction in acoustics.}

Indeed, by the physical meaning of $h$ one can prepare small particles 
with impedances $\zeta = \frac{h(x,\om)}{a}$, depending on the frequency
$\om$ in a desired way. 

Let us show that equation \eq{3} can be solved
non-uniquely with respect to the three functions: $h_1(x,\om)$,
$h_2(x,\om)\leq 0$, and $N(x)\geq 0$, given two arbitrary functions
$p_1(x,\om)$ and $p_2(x,\om)\geq 0$. We emphasize that $N(x)$ can be
chosen so that  {\em it does not depend on $\om$}.
\begin{lem}\label{L1} For any $p_1(x,\om)$ and $p_2(x,\om)\leq 0$, there
exist three real-valued functions $h_1(x,\om)$, $h_2(x,\om)\leq 0$, and 
$N(x)\geq 0,$ such that
equation \eq{3} is satisfied. The choice of these functions $h_1$, 
$h_2\leq 0$ and
$N(x)\geq 0$ is non-unique. 
\end{lem}
\begin{proof}
Our proof yields a method for finding a triple $\{ h_1,h_2,N(x)\}   
$ given the tuple $\{p_1,p_2\}$, and analytic formulas for
the functions $h_1$, $h_2$, and $N(x)$, see formulas (7), (13),
(14) and (15).  

Denote 
$$R:= (p_1^2 + p_2^2)^{1/2},$$ 
and let $\psi$ be the argument of $p= p_1 + ip_2$, so that
\begin{equation}
\label{E4} 
p_1 = R\cos \psi,\quad p_2 = R\sin \psi; \quad R=R(x,\omega).
\end{equation}
Equation \eq{3} is equivalent to the following:
\begin{eqnarray}
\label{E5}
\frac{4\pi N(x)[h_1(1 + h_1) + h_2^2]}{(1 + h_1)^2 + h_2^2}&= &R\cos \psi,\\
\label{E6}
\frac{4\pi N(x) \, h_2}{(1 + h_1)^2 + h_2^2}&=&R\sin \psi.
\end{eqnarray}
Let 
\begin{equation}
\label{E7}
(1 + h_1)^2 + h_2^2:= r^2,\quad h_2 = r\sin \vp,\quad 1 + h_1 = r\cos \vp.
\end{equation}
Equations \eq{5} -- \eq{6} can be written as:
\begin{eqnarray}
\label{E8}
4\pi N(x)[1-\frac{1}{r}\, \cos \vp] &=& R \, \cos \psi,\\
4\pi N(x) \frac{\sin \vp}{r} & = & R\, \sin \psi.\label{E9}
\end{eqnarray}
Let 
\begin{equation}
\label{E10}
\frac{R}{4\pi N(x)}:= \begin{cases} \rho(x,\om) & \mbox{if } R(x,\om) >
0,\\
0 & \mbox{if } R(x,\om) = 0. 
\end{cases}
\end{equation}
Then 
\begin{eqnarray}
\label{E11} - \frac{\cos \vp}{r} &= & \rho \, \cos \psi - 1, \\
\label{E12} \frac{\sin \vp}{r} & = & \rho \, \sin \psi.
\end{eqnarray}
Given $R$ and $\psi$, we want to find $r,\vp$ and $N(x)$ so that
equations \eq{10} -- \eq{12} hold.
\par
From \eq{11} -- \eq{12} one gets:
\begin{equation}
\label{E13}
r = \frac{1}{\sqrt{\rho^2 - 2\rho\, \cos \psi + 1}}\,.
\end{equation}
If $r$ is found by formula \eq{13}, one finds 
\begin{equation}
\label{E14}
\sin \vp = \frac{\rho \, \sin \psi}{\sqrt{\rho^2-2\rho\, \cos \psi +
1}}\,,\quad \cos \vp = \frac{1-\rho \, \cos \psi}{\sqrt{\rho^2-2\rho\,
\cos \psi + 1}}\,.
\end{equation}
By definition \eq{10} $\rho \geq 0$. Choosing $N(x)$ so that $\rho < 1$,
one can make sure that $1-2\rho\, \cos \psi + \rho^2 > 0$ and $\cos
\vp > 0$. The choice of $N(x)\geq 0$ is made independently of $\om$. 

One has:
\begin{equation}
\label{E15}
R(x,\om) = 4\pi N(x)\, \rho(x,\om).
\end{equation}
If $p_2\leq 0$, which corresponds to some absorption in the material,
then $\sin \psi \leq 0$, $h_2 \leq 0$, and $\sin \vp \leq 0$. If $N(x)$ is
sufficiently large, so that $\rho < 1$, then \eq{14} implies $\cos \vp >
0$ and $\sin \vp \leq 0$, so $-\frac{\pi}{2} < \vp \leq 0$. One can get
the dispersion relation for $\om(k)$ from the equation
\begin{equation}
\label{E16}
\frac{\om^2}{c^2}\, n^2(x,\om)=k^2,
\end{equation}
where $c = {\rm const} $ is the speed of the waves in the free space.
\par
From \eq{16} one gets
\begin{equation}
\label{E17}
\om n(x,\om) = \n k\n c\,.
\end{equation}
Taking gradient with respect to $k$, one gets
\begin{equation}
\label{E18}
\bka n(x,\om) + \om\, \frac{\pa n}{\pa \om}\bkz \na_k\om = ck^0,
\quad k^0 = \frac{k}{\n k\n} \,.\end{equation}
Therefore the group velocity $$v_g = \na_k\om$$ will be directed opposite
to $k^0$ if
\begin{equation}
\label{E19}
\frac{\om}{n}\: \frac{\pa n}{\pa \om}< -1.
\end{equation}
In deriving \eq{19} we assume that $n(x,\om)$ is a real-valued function,
or that the imaginary part of $n(x,\om)$ in $D$ is negligible, so that
when the wave is propagating through $D$ its amplitude changes just a
little. For this to happen it is sufficient to assume that $L\n \Im
n\n\ll 1$, where $L$ is the diameter of $D$.
\par
Since an arbitrary function  $n(x,\om)$ can be created  by the method of
\cite{R529}, one can ensure that inequality \eq{19} holds. This implies
that the created material has negative refraction in the sense that the
group velocity vector is directed opposite to the wave vector $k$.
\end{proof}
{\it Example.} If $n(x,\om)=(1+c\om^2)^{-1}$, where $c>0$
is a constant, then inequality \eq{19} holds if $c\om^2>1$.
Suppose that the frequency interval is $[\om_{min}, \om_{max}]$.
Then inequality \eq{19} holds for all frequencies from this interval
if $c>\frac 1{\om_{min}^2}$. 

{\it Let us give a summary of our method for creating
material with the desired refraction coefficient:}

{\it Step 1.} Given the desired refraction coefficient $n^2(x,\om)$ in 
$D$, and the 
original  refraction coefficient $n_0^2(x,\om)$, 
calculate the function
$p(x,\om)=k^2[n_0^2(x,\om)-n^2(x,\om)]$.

{\it Step 2.} Given $p(x,\om)$, solve eq. (3) for 
$h(x,\om)=h_1(x,\om)+ih_2(x,\om)$,
$h_2(x,\om)\leq 0,$ and $N(x)\geq 0$.

There are many solutions: the three real-valued functions
$h_1, h_2\leq 0,$ and $N(x)\geq 0$ are to be found from the two
real-valued functions $p_1$ and $p_2\leq 0$.

In this step the solutions can be calculated
analytically.

{\it Step 3.} After $h(x,\om)$ and $N(x)$ are found in Step 2, one
partitions the domain $D$ into a union of small
non-intersecting cubes $Q_j$, centered at the points $x_j$, and
embeds in each $Q_j$ the number  $\nu_j:=[\int_{Q_j}N(x)dx/a]$ of small
balls of radius $a$, each having the boundary parameter
$\zeta_j=h(x_j,\om)/a$, where the symbol $[b]$ denotes the closest integer
to the number $b\geq 0$.

\end{document}